\newif\ifready\readyfalse
\newif\ificalp\icalpfalse
\newif\iftpdp\tpdpfalse

\documentclass{article}
\usepackage[utf8]{inputenc}
\usepackage{graphicx}
\usepackage{booktabs}
\usepackage[top=1in,left=1in,right=1in,bottom=1in]{geometry}
\usepackage{amsthm}
\usepackage{amsmath}
\usepackage{dsfont}
\usepackage[normalem]{ulem}
\usepackage{multirow}
\usepackage{diagbox}

\usepackage{amssymb}
\usepackage[dvipsnames]{xcolor}
\usepackage[colorlinks=true,linkcolor=blue,citecolor=blue,allcolors=blue,hypertexnames=false]{hyperref}
\hypersetup{
    colorlinks=true,
    linkcolor=blue,
    filecolor=blue,
    urlcolor=blue,
    allcolors=blue
}
\usepackage{nicematrix}

\usepackage{relsize,stmaryrd}
\usepackage{algorithmicx}
\usepackage{algorithm}
\usepackage[noend]{algpseudocode}
\usepackage{multicol}
\usepackage[most]{tcolorbox}
\usepackage{pgfplots}
\usepackage{subcaption}
\pgfplotsset{
    compat=1.3,
    legend image code/.code={
        \draw [#1] (0cm,-0.1cm) rectangle (0.6cm,0.1cm);
    },
}

\usepackage{thm-restate}
\usepackage{enumerate}
\usepackage{enumitem}
\setlist{noitemsep,topsep=0pt,parsep=0pt,partopsep=0pt}
\usepackage{mathtools}
\usepackage{xspace}
\usepackage{orcidlink}

\usepackage[capitalize,nameinlink]{cleveref}

\theoremstyle{plain}
\newtheorem{theorem}{Theorem}[section]

\newtheorem{definition}[theorem]{Definition}

\newcommand{\defn}[1]{\textbf{\textit{#1}}}

\newtheorem{thm}{Theorem}[section]
\newtheorem{lem}[thm]{Lemma}
\newtheorem{defnt}[thm]{Definition}
\newtheorem{clm}[thm]{Claim}

\crefname{thm}{Theorem}{Theorems}
\crefname{lem}{Lemma}{Lemmas}
\crefname{def}{Definition}{Definitions}
\Crefname{clm}{Claim}{Claims}
\crefname{theorem}{Theorem}{Theorems}
\Crefname{lemma}{Lemma}{Lemmas}
\Crefname{claim}{Claim}{Claims}
\Crefname{observation}{Observation}{Observations}
\Crefname{algorithm}{Algorithm}{Algorithms}
\Crefname{myalgctr}{Algorithm}{Algorithms}
\Crefname{challenge}{Challenge}{Challenges}

\algrenewcommand\algorithmicindent{1em}%

\newcommand{\alg}{\mathcal{A}}

\newcommand{\var}[1]{\mathrm{Var} \Big[ #1 \Big]}
\newcommand{\V}{\mathrm{Var}}

\newcommand{\E}[0]{\ensuremath{\mathbb{E}}}

\definecolor{mygreen}{RGB}{20,140,80}
\definecolor{linkcolor}{RGB}{0,0,230}
\definecolor{mylightgray}{RGB}{230,230,230}
\definecolor{verylightgray}{RGB}{245,245,245}

\newcommand{\etal}[0]{et al.}

\newcounter{myalgctr}

\newtcolorbox{OuterBox}[1][]{%
    breakable,
    enhanced,
    frame hidden,
    interior hidden,
    left=-5pt,
    right=-5pt,
    top=-5pt,
    float=p,
    boxsep=0pt,
    arc=0pt
#1}%

\newtcolorbox{InnerBox}[1][]{%
    enforce breakable,
    enhanced,
    colback=gray,
    colframe=white,
#1}%

\newenvironment{tbox}{
\vspace{0.2cm}
\begin{tcolorbox}[width=\columnwidth,
                  enhanced,
                  boxsep=2pt,
                  left=1pt,
                  right=1pt,
                  top=4pt,
                  boxrule=1pt,
                  arc=0pt,
                  colback=white,
                  colframe=black,
	              breakable
                  ]%
}{
\end{tcolorbox}
}

\newcommand{\tboxhrule}[0]{\vspace{0.1cm} {\color{black} \hrule} \vspace{0.2cm}}

\newenvironment{titledtbox}[1]{\begin{tbox}#1 \tboxhrule}{\end{tbox}}

\newcommand{\drawFigTwo}{1}

\newcommand{\Adj}{A}

\newcommand{\ledp}{LEDP\xspace}

\newcommand{\naturals}{\mathbb{N}}

\newcommand{\eps}{\varepsilon}

\newcommand{\expect}{\mathbb{E}}

\newcommand{\geom}{\mathsf{Geom}}
\newcommand{\rr}{\mathsf{RandomizedResponse}}

\newcommand{\logfactor}{\log^2 n}

\newcommand{\df}{\sens f}
\newcommand{\sens}{\Delta}

\newcommand{\draw}{\geom(\exp(\eps/\logfactor))}

\ifready\else
    
    \newcommand{\qqnote}[1]{\footnote{\color{magenta} Quanquan: #1}}
    \newcommand{\talya}[1]{{\color{ForestGreen} #1}}

    \newcommand{\told}[1]{\sout{#1}}
    \newcommand{\snote}[1]{\footnote{\color{blue} Sofya: #1}}
    
    \newcommand{\anote}[1]{\footnote{\color{orange} Adam: #1}}

\fi
\ifready
    
    \newcommand{\qqnote}[1]{}
    \newcommand{\talya}[1]{{#1}}
    \newcommand{\told}[1]{}
    
    \newcommand{\snote}[1]{}
    
    \newcommand{\anote}[1]{}
\fi

\newcommand{\cfours}{C_4}
\newcommand{\That}{\widehat{T}}

\newcommand{\domain}{\mathcal{D}}
\newcommand{\reals}{\mathbb{R}}
\newcommand{\norm}[1]{\left\Vert#1\right\Vert}

\newcommand{\dmax}{d_{max}}

\newcommand{\range}{Range}
\newcommand{\rangeout}{\mathcal{Y}}

\newcommand{\cB}{\mathcal{B}}
\newcommand{\cC}{\mathcal{C}}
\newcommand{\cP}{\mathcal{P}}
\newcommand{\adj}{a}
\newcommand{\adjb}{\mathbf{a}}

\newcommand{\indicator}{\mathds{1}}

\newcommand{\paren}[1]{{\left( {#1} \right)}}
\newcommand{\bit}{\{0,1\}}

\usepackage{pstricks,tikz}
\usetikzlibrary{decorations.markings}
\usetikzlibrary{patterns}
\usetikzlibrary{calc, intersections}
\usepackage{pgf,tikz}
\usetikzlibrary{calc}
\usepackage{tkz-euclide}
\usepackage{pgfplots}
\usetikzlibrary{matrix}
\usepackage{nicematrix}

\tikzset{dotted pattern/.style args={#1}{
		postaction=decorate,
		decoration={
			markings,
			mark=between positions 0.25 and 0.75 step 0.25 with {
				\fill[radius=#1] (0,0) circle;
			}
		}
	},
	dotted pattern/.default={1pt},
}

\newcommand{\interLB}{

	\begin{scope}[xscale=7, yscale=3.5]
				\foreach \x\i\l in {0/1/1,.25/2/2,.5/3,1.15/4/n} {
					\node[draw,circle,inner sep=0.1cm, fill=black, label={above:$v_{\l}$}
					] (t\i) at (\x,1) {};
				}
			
			\foreach \x\i\l\r in {-.1/1/1/2,.3/2/3/4, 1/3/2n-1/2n} {
				\node[draw,circle,inner sep=0.1cm, fill=blue, label={below:$\l$}
				] (b\i) at (\x,0) {};
				\node[draw,circle,inner sep=0.1cm, fill=blue, label={below:$\r$}
				] (b'\i) at (\x+.2,0) {};
				\draw[dotted] (b\i) -- (b'\i); 
			}
		\foreach \i in {1,2,3} {
			\foreach \j in {1,2,3,4} {
				\draw[very thin, black] (b\i) -- (t\j);
				\draw[very thin, black] (b'\i) -- (t\j);
			}
		}
		\path[dotted pattern] (.65,1) -- (1,1);
		\path[dotted pattern] (.7,0) -- (.85,0);	
		\node[] (V1) at (-.3,1.05) {\large $V_1$};
		\node[] (V1) at (-.3,0.05) {\large $V_2$};
				
			\end{scope}
}

\title{Triangle Counting with Local Edge  Differential Privacy}
\author{Talya Eden\thanks{Bar Ilan University, talyaa01@gmail.com}, Quanquan C. Liu\thanks{Northwestern University, quanquan@northwestern.edu}, Sofya Raskhodnikova\thanks{Boston University, \{sofya, ads22\}@bu.edu}, Adam Smith\footnotemark[3]}

\author{%
    \begin{tabular}{cc}	
        \begin{tabular}{c}
            Talya Eden\\
            Bar-Ilan University\\
            \texttt{talyaa01@gmail.com}
            \orcidlink{https://orcid.org/0000-0001-8470-9508}
        \end{tabular}
        & 	
        \begin{tabular}{c}
            Quanquan C. Liu\\
            Yale University\\
            \texttt{quanquan.liu@yale.edu}
            \orcidlink{https://orcid.org/0000-0003-1230-2754}
        \end{tabular}
        \\
        \\
        \begin{tabular}{c}
            Sofya Raskhodnikova\\
            Boston University\\
            \texttt{sofya@bu.edu}
            \orcidlink{https://orcid.org/0000-0002-4902-050X}
        \end{tabular}
         & 
         \begin{tabular}{c}
            Adam Smith\\
            Boston University\\
            \texttt{ads22@bu.edu}
            \orcidlink{https://orcid.org/0000-0001-9393-1127}
        \end{tabular}
    \end{tabular}
}
\date{}

\begin{document}
\renewcommand{\algorithmicrequire}{\textbf{Input:}}
\renewcommand{\algorithmicensure}{\textbf{Output:}}
\algblock{ParFor}{EndParFor}
\algblock{Input}{EndInput}
\algblock{Output}{EndOutput}
\algblock{ReduceAdd}{EndReduceAdd}

\algnewcommand\algorithmicparfor{\textbf{parfor}}
\algnewcommand\algorithmicinput{\textbf{Input:}}
\algnewcommand\algorithmicoutput{\textbf{Output:}}
\algnewcommand\algorithmicreduceadd{\textbf{ReduceAdd}}
\algnewcommand\algorithmicpardo{\textbf{do}}
\algnewcommand\algorithmicendparfor{\textbf{end\ input}}
\algrenewtext{ParFor}[1]{\algorithmicparfor\ #1\ \algorithmicpardo}
\algrenewtext{Input}[1]{\algorithmicinput\ #1}
\algrenewtext{Output}[1]{\algorithmicoutput\ #1}
\algrenewtext{ReduceAdd}[2]{#1 $\leftarrow$ \algorithmicreduceadd(#2)}
\algtext*{EndInput}\input
\algtext*{EndOutput}
\algtext*{EndIf}
\algtext*{EndFor}
\algtext*{EndWhile}
\algtext*{EndParFor}
\algtext*{EndReduceAdd}

\maketitle
\sloppy

\begin{abstract}
Many deployments of differential privacy in industry are in the local model, where each party releases its private information via a differentially private randomizer. 
We study triangle counting in the local model with edge differential privacy (that, intuitively, requires that the outputs of the algorithm on graphs that differ in one edge be  indistinguishable). In this model, each party's local view consists of the adjacency list of one vertex. We investigate both noninteractive and interactive variants of the model.

In the noninteractive model, we prove that additive $\Omega(n^2)$ error is necessary {for sufficiently small constant $\eps$, where $n$ is the number of nodes and $\eps$ is the privacy parameter}. This lower bound is our main technical contribution. It uses a reconstruction attack with a new class of linear queries and a novel mix-and-match strategy of running the local randomizers with different completions of their adjacency lists. It matches the additive error of the algorithm based on Randomized Response, proposed by Imola, Murakami and Chaudhuri (USENIX2021) and analyzed by Imola, Murakami and Chaudhuri (CCS2022) for constant $\eps$.
We use a different postprocessing of Randomized Response and provide tight bounds on the variance of the resulting algorithm.

In the interactive setting, we prove a lower bound of $\Omega(n^{3/2}/\eps)$ on the additive error for $\eps\leq 1$. Previously, 
no hardness results were known for interactive, edge-private algorithms in the local model, except for those that follow trivially from the results for the central model. Our work significantly improves on the state of the art in differentially private 
graph analysis
in the local model.
\end{abstract}

\section{Introduction}\label{sec:intro}
Triangle counting is a fundamental primitive in graph analysis, used in numerous applications and widely studied in different computational models~\cite{pagh2012colorful,ParkSKP14,ELRS15,McGregorVV16,KallaugherP17,AKK19,MPC-small,McGregorV20,chen2021triangle}.
Statistics based on triangle counts reveal important structural information about networks (as discussed, e.g., in \cite{PDBL16,farkas2001spectra,milo2002network}). They are used to perform many computational tasks on social networks, including  community detection~\cite{palla2005uncovering}, link prediction~\cite{eckmann2002curvature}, and spam filtering~\cite{BecchettiBCG08}.
See~\cite{al2018triangle} for a %
survey on algorithms for and applications of triangle counting.

In applications where a graph (e.g., a social network) holds sensitive information, the algorithm that computes on the graph has to protect personal information, such as friendships between specific individuals. Differential privacy~\cite{DMNS06} has emerged as the standard of rigorous privacy guarantees. See \cite{RaskhodnikovaS16} for a survey of differentially private graph analysis. The most investigated setting of  differential privacy is called the {\em central model}. It implicitly assumes a curator that collects all the data, performs computations on it, and provides data releases. In some situations, however, it might be undesirable to collect all information in one place, for instance, because of trust or liability issues. To address this, the {\em local model} of differential privacy was proposed~\cite{EvfimievskiGS03,DMNS06,kasiviswanathan2011can} and is now used in many industry deployments~\cite{rappor14,BMMR17,CJKL18,2017LearningWP,DKY17}.

In this model, each party releases its private information via  a differentially private randomizer. Then the algorithm processes the information and, in the case of the local \emph{noninteractive} model, outputs the answer. In the case of the local \emph{interactive} model, the algorithm may have multiple rounds where it asks all parties to run different randomizers on their private data. These randomizers can have arbitrary dependencies on previous messages. Differential privacy in the local model is defined with respect to the whole transcript of interactions between the parties and the algorithm.
In the local model applied to graph data, each vertex represents a party. It receives the list of its neighbors as input and applies local randomizers to it. In contrast to the typical datasets, where information belongs to individual parties, in the graph setting, each pair of parties (vertices) share the information of whether there is an edge between them.

Differential privacy, intuitively, guarantees that, for any two neighboring datasets, the output distributions of the algorithm are roughly the same. There are two natural notions of neighboring graphs: {\em edge-neighboring} and {\em node-neighboring}. Two graphs are edge-neighboring if they differ in one edge; they are node-neighboring if they differ in one node and its adjacent edges. Edge differential privacy is, in general, easier to attain, but node differential privacy provides stronger %
guarantees. Edge differential privacy was introduced and first applied to triangle counting in~\cite{NissimRS07}. The edge-differentially private algorithm from~\cite{NissimRS07} was generalized and implemented in~\cite{KarwaRSY14}.
The first node-differentially private algorithms appeared in~\cite{BlockiBDS13,KasiviswanathanNRS13,ChenZ13}, and all three of these articles considered the problem of triangle counting. 
Edge differential privacy in the local model has been studied in~\cite{Qin17,gao2018local,sunAnalyzing2019,Ye0A0X20,ImolaMC21,ImolaMC22,DLRSSY22,hillebrand2023communication,liu2024edge}
with most of the listed articles focusing on triangle counting.

In this work, we investigate edge differentially private algorithms for estimating the number of triangles in a graph in the local model. Our goal is to understand the additive error achievable by such algorithms  both in the noninteractive and in the interactive model. 
For the noninteractive model, we   provide  upper and lower bounds on additive error.   Our bounds are tight in terms of $n,$ the number of nodes in the input graph.  
For the interactive model, we provide the first lower bound specific to 
local, edge differentially  private (LEDP) algorithms. There are easy lower bounds for the central model (based on global sensitivity), which a fortiori
apply to the local model. 
However, no lower bounds specific to the local model were previously known  for any graph problem, even for 2-round algorithms.
Together, our results improve our understanding of both noninteractive and interactive LEDP algorithms.

\subsection{Results}
Our results and comparison to previous work are summarized in~\cref{tab:results}.

\begin{table}[htb!]
\renewcommand\arraystretch{2}
    \centering
    \begin{tabular}{| c | c || c | c |c|c|}
         \hline
         \multicolumn{2}{|c||}{Model} &  \multicolumn{2}{|c||}{Previous Results }& \multicolumn{2}{|c||}{Our Results}  \\
         \hline
         \hline
          \multirow{2}*{\shortstack{Non-\\interactive}} & Lower Bound &   $\Omega(n^{3/2})$& \cite{ImolaMC21} & $\Omega(n^2)$ & \ificalp
\else
          \cref{thm:noninteractiv-lb-intro}
          \fi\\ \cline{2-6}
          & Upper Bound &  
          $O(n^2)$ (constant $\eps$) & \cite{IMCshuffler22} & $O\Big(\frac{\sqrt{C_4(G)}}{\eps}+\frac{n^{3/2}}{\eps^3}\Big)$ & 
          \ificalp
\else
          \cref{thm:noninteractiv-ub-intro}
          \fi\\
          \specialrule{.2em}{.1em}{.1em} 

          \multirow{2}*{Interactive} & Lower Bound &  
          $\Omega(n)$ & easy & $\Omega(\frac{n^{3/2}}\eps)$ & 
          \ificalp
\else
          \cref{thm:interactive-lb-intro}
          \fi
           \\ \cline{2-6}
          & Upper Bound & 
          $O\left(\frac{\sqrt{\cfours(G)}}{\eps}+\frac{n^{3/2}}{\eps^2}\right)$
          & \cite{ImolaMC22} & \multicolumn{2}{|c}{} 
          \\
          \cline{1-4}
    \end{tabular}
    \caption{Summary of lower and upper bounds on the additive error for triangle counting in the noninteractive and interactive models. Note that 
    the largest value of $\cfours(G)$ is $\binom{n}{4}=\Theta(n^4).$  
    For ease of comparison, the results of \cite{ImolaMC21} and \cite{ImolaMC22} are stated for graphs with $d_{max}=\Theta(n)$. 
    }
    \label{tab:results}
\end{table}

\subsubsection{Lower Bound for the Noninteractive Local Model} Our main technical contribution is a lower bound in the noninteractive setting. It uses a reconstruction attack (for the central model) with a new class of linear queries and a novel mix-and-match strategy of running local randomizers with different completions of their adjacency lists. While reconstruction attacks are a powerful tool in proving lower bounds in the central model of differential privacy, they have not been used to obtain bounds in the local model. Previous lower bounds in the local model  are based on quite different techniques---typically, information-theoretic arguments (see, for example,  \cite{kasiviswanathan2011can,beimel2008distributed,DuchiJW13} and many subsequent works).  

\begin{thm}[Noninteractive Lower Bound, informal version]\label{thm:noninteractiv-lb-intro}
Let  $\eps\in(0,1/20)$ and $\delta\geq0$ be a sufficiently small constant. 
There exists a family of graphs such that every noninteractive 
$(\eps,\delta)$-local edge differentially private algorithm 
that gets an $n$-node graph from the family as input and approximates the number of triangles in the graph within additive error at most $\alpha$ (with sufficiently high constant probability) must have $\alpha=\Omega(n^2)$.
\end{thm}
Our lower bound holds for all small $\delta\geq 0$ (the case referred to as ``approximate'' differential privacy). Observe that such lower bounds are stronger than those for $\delta=0$ (the case referred to as ``pure'' differential privacy), because they include $\delta=0$ as a special case. 
The only previously known lower bound, due to Imola et al.~\cite{ImolaMC21}, showed that noninteractive algorithms must have error 
$\Omega(\sqrt n\cdot d_{max})$.

To prove the lower bound in~\cref{thm:noninteractiv-lb-intro}, we develop a novel mix-and-match technique for noninteractive local model. For a technical overview of the proof of \cref{thm:noninteractiv-lb-intro}, see \cref{sec:high-level-lb-intro}.

Our lower bound matches the upper bound of   $O(n^2)$   proved by \cite[Theorem G.3]{IMCshuffler22} (for constant $\eps$) for an algorithm based on randomized response. In this work, we give a simpler variant of the algorithm and a more refined analysis, which works for all $\eps.$

\subsubsection{Tight Analysis of Randomized Response}

The most natural algorithm for the noninteractive model is Randomized Response, which dates back to Warner~\cite{Warner65}. In this algorithm, each bit is flipped with probability $\frac 1{e^{\eps}+1},$ where $\eps$ is the privacy parameter. In the case of graphs, each bit represents the presence or absence of an edge. An algorithm based on Randomized Response for triangle counting was first analyzed by \cite{ImolaMC21} 
for the special case of Erd\H{o}s-R\'enyi graphs, and then \cite{IMCshuffler22} proved that this algorithm has $O(n^2)$ additive error for constant $\eps$ for general graphs. These works first compute the number of triangles and other induced subgraphs with three vertices as though the noisy edges are real edges and then appropriately adjust the estimate using these counts to make it unbiased.

We use a different postprocessing of Randomized Response {than in \cite{ImolaMC21,IMCshuffler22}}. %
 We %
rescale the noisy edges right away, so we need not compute counts for graphs other than triangles, which makes the analysis much simpler.
{Note that the output distribution of our algorithm is the same as the output distribution of the algorithm in \cite{ImolaMC21,IMCshuffler22}, even though our algorithm performs a different computation.}
We obtain tight upper and lower bounds on the variance of the resulting algorithm that hold for all $\eps.$
Our bounds are more refined, as they are stated in terms of $C_4(G)$, the number of four cycles in the graph.

\begin{thm}[Analysis of Randomized Response]\label{thm:noninteractiv-ub-intro}
    For all $\eps \in (0,1]$,
    there exists a noninteractive
     $\eps$-LEDP algorithm  based on Randomized Response that
gets an $n$-node graph as input and 
     returns an estimate $\hat{T}$ of the number of triangles in the graph
     that has variance $\Theta\Big(\frac{C_4(G)}{\eps^2}+\frac{n^{3}}{\eps^6}\Big)$. 
   
     In particular, with high constant probability, $\hat{T}$ has additive error $\alpha=O\Big(\frac{\sqrt{C_4(G)}}{\eps}+\frac{n^{3/2}}{\eps^3}\Big).$
\end{thm}

Note that for  constant $\eps$, both the result of \cite{IMCshuffler22} and \cref{thm:noninteractiv-ub-intro} give an upper bound of $O(n^2)$ on the additive error of the algorithm's estimate. 
Thus, Randomized Response is optimal for graphs that have $\cfours=\Theta(n^4)$ by our lower bound in~\cref{thm:noninteractiv-lb-intro}. Also, observe that Randomized Response achieves pure differential privacy (with $\delta=0$), whereas the lower bound in~\cref{thm:noninteractiv-lb-intro} holds even for approximate differential privacy. Even though allowing $\delta>0$ results in better accuracy for many problems, it does not give any additional utility for noninteractive triangle counting.

\subsubsection{Lower Bound for the Interactive Local Model} 

Next, we investigate triangle counting in the interactive setting. 
Imola et al.~\cite{ImolaMC21} present an $\eps$-LEDP for triangle counting in the interactive model with additive  error of
$O\big(\sqrt{\cfours(G)}/{\eps}+{\sqrt n\cdot \dmax}/{\eps^2}\big)$,
where $\dmax$ is an upper bound on the maximum degree.

We give a lower bound on the additive error of LEDP algorithms  for triangle counting in the interactive model. Note that $\Omega(n)$ additive error is unavoidable for triangle counting even in the central model, because the (edge) global sensitivity of the number of triangles is $n-2$ (and this lower bound is tight in the central model).
There were no previously known lower bounds for this problem (or any other problem on graphs) specific to the interactive LEDP model that applied to even 2-round algorithms. Our lower bound applies to interactive algorithms with \emph{any} number of rounds.

\begin{thm}[Interactive Lower Bound]\label{thm:interactive-lb-intro}
There exist a family of graphs and a constant $c>0$ such that for every %
$\eps  \in (0,1)$,  $n \in \naturals, \alpha\in(0,n^2]$ and  $\delta \in \left[0, \frac{1}%
{{500}}
\cdot \frac{\eps^3 \alpha^2}{n^5 \ln(n^3/(\eps \alpha))}\right]$, 
every (potentially interactive) 
$(\eps,\delta)$-local edge differentially private algorithm 
that gets an $n$-node graph from the family as input and  approximates the number of triangles in the graph
with additive error at most $\alpha$ (with  probability at least 2/3) must have $\alpha\geq c\cdot\frac{n^{3/2}}\eps$.
\end{thm}
Our lower bound is obtained via a  reduction from the problem of computing the summation of $n$ randomly sampled bits in $\{0,1\}$ in the LDP model, studied in a series of works~\cite{beimel2008distributed,chan2012optimal,duchi2018minimax,JosephMNR19}.
Our lower bound matches the upper bound of \cite{ImolaMC21} 
for constant $\eps$ and for graphs where $d_{max}=\Theta(n)$ and $C_4(G)=O(n^3)$. It is open whether additive error of $o(n^2)$ can be achieved for general graphs.

\medskip

\subsection{Technical Overview of the Noninteractive Lower Bound}\label{sec:high-level-lb-intro}

Typical techniques for proving lower bounds in the local model heavily rely on two facts that hold for simpler datasets: first, each party's information is not seen by other parties; second, arbitrary changes to the information of one party have to be protected. Both of these conditions fail for graphs in the LEDP model: each edge is shared between two parties, and only changes to one edge are protected in the strong sense of neighboring datasets, imposed by differential privacy.

To overcome these difficulties, we develop a new lower bound method, based on reconstruction attacks in the central model.
Such attacks use accurate answers to many queries to reconstruct nearly all the entries of a secret dataset~\cite{DiNi03,DworkMT07,DworkY08,KasiviswanathanRSU10,KasiviswanathanRS13,anindya_lb}. They are usually applied to algorithms that release many different values. However, a triangle-estimation algorithm returns a single number. 
Consider a  naïve attempt to mount an attack using the algorithm as a black box, that is, by simulating every query using a separate invocation of the triangle counting algorithm. This would require  us to run the local randomizers many times, degrading their privacy parameters and making a privacy breach vacuous. 

To resolve this issue, in our attack, we use the noninteractive
triangle-estimation algorithm as a {\em gray box}. 
Since the algorithm is noninteractive, it is specified by local randomizers for all vertices and a postprocessing algorithm that runs on the outputs of the randomizers. We use a secret dataset $X$ to create a secret subgraph, run the randomizers for the vertices in the secret subgraph only twice, and publish the results. By properties of the randomizers and by composition, the resulting procedure is differentially private. In the next phase, we postprocess the published information to complete the secret subgraph to different graphs corresponding to the queries needed for our attack. Then we feed these graphs to the triangle approximation algorithm, except that for the vertices in the secret subgraph, we rely only on the published outputs. If the triangle counting algorithm is accurate, we get accurate answers to our queries. Even though the randomness used to answer different queries is correlated, we show that a good approximation algorithm for triangle counts allows us to get most of the queries answered correctly. 
Finally, we use a novel anti-concentration bound (\cref{lem:anti-concentration}, below)
to demonstrate that our attack succeeds in reconstructing most of the secret dataset with high probability. This shows that the overall algorithm we run in this process is not differentially private, leading to the conclusion that a very accurate triangle counting algorithm cannot exist in the  noninteractive LEDP setting.

We call the  queries used in our attack  {\em outer-product queries}. The queries are linear, but their entries are dependent. To define this class of queries, we represent the secret dataset $X$ with $n^2$ bits as an $n\times n$ matrix. An  outer-product query to $X$ specifies two vectors $A$ and $B$ of length $n$ with entries in $\{-1,1\}$ and returns $A^TXB,$ that is, $\sum_{i,j\in[n]} A_iX_{ij} B_j$.

To analyze our reconstruction attack, we prove the following anti-concentration bound
for random outer-product queries,
which might be of independent interest. 

\begin{lem}[Anti-concentration for random outer-product queries]\label{lem:anti-concentration}
Let $M$ be an $n\times n$ matrix with entries $M_{ij}\in\{-1,0,1\}$ for all $i,j\in[n]$ and $m$ be the number of nonzero entries in $M.$
Let $A$ and $B$ be drawn uniformly and independently from $\{-1,1\}^n.$ If $m\geq\gamma n^2$ for some constant $\gamma$, then
$$\Pr\Big[|A^T M B|>\frac{\sqrt{m}}2\Big]\geq \frac{\gamma^2}{16}.
$$
\end{lem}

The literature on reconstruction attacks describes other classes of dependent queries~\cite{KasiviswanathanRSU10}; the outer-product queries arising here required a new and qualitatively different analysis.

\subsection{Additional Related Work}\label{app:related}
One of the difficulties with proving lower bounds in the local model is that Randomized Response, despite providing strong privacy guarantees, supplies enough information to compute fairly sophisticated statistics.
For example, Gupta, Roth and Ullman~\cite{GuptaRU12} show how the output of Randomized Response can be used to estimate the density of all cuts in a graph. Karwa et al.~\cite{KarwaSK2014-ergm} show how to fit exponential random graph models based on randomized response output. For certain model families, this would entail estimation of the number of triangles; however, they provide no theoretical error analysis, only experimental evidence for convergence. Randomized Response has also been studied in the statistics literature with a focus on small probabilities of flipping an edge. Balachandran et al.~\cite{BalachandranKV2017-low-error} analyze the distribution of the naive estimator that counts the number of triangles in the randomized responses (when flip probabilities are very low). Chang et al.~\cite{ChangKY2022-multi-replicates} give estimation strategies for settings where the flip rate is unknown but multiple replicates with independent noise are available. To the best of our understanding, these works do not shed light on the regime most relevant to privacy, where edge-flip probabilities are close to $1/2$.

A number of  works have looked at triangle counting and other graph problems
in the empirical setting~\cite{sunAnalyzing2019,Qin17,gao2018local,Ye0A0X20} in ``decentralized'' privacy models. In all but~\cite{sunAnalyzing2019}, the local view consists of the adjacency list. 
The local views in Sun \etal~\cite{sunAnalyzing2019} consist of two-hop neighborhoods. Such a model results in less error since nodes 
see all of their adjacent triangles and can report how many they see using the geometric mechanism.

\subsection{Organization}
Various models of differential privacy, including LEDP, are defined in \cref{sec:privacy_defs}. Our proof of the lower bound for the noninteractive model,  \cref{thm:noninteractiv-lb-intro}, appears in 
\cref{sec:one-round-lower}. The anti-concentration lemma for out-product queries (\cref{lem:anti-concentration}) is proved in \cref{sec:anti-concentration}.
Our analysis of   Randomized Response and the proof of \cref{thm:noninteractiv-ub-intro}
appears in 
\ificalp
\else
\cref{sec:one-round-upper}.
\fi
The proof of \cref{thm:interactive-lb-intro} for the interactive LEDP model appears in \cref{sec:interactive-lb}.
\ificalp
\else
In \cref{sec:append-privacy-tools}, we state several privacy tools (\cref{sec:privacy-tools})  
and concentration bounds (\cref{sec:concentration-bounds}).
\fi

\section{Background on Differential Privacy}\label{sec:privacy_defs}

We begin with the definition of differential privacy that applies to datasets represented as vectors as well as to graph datasets. 

\begin{defnt}[Differential Privacy \cite{DMNS06,DworkKMMN06}]\label{def:dif-privacy} 
Let $\eps>0$ and $\delta\in[0,1)$. A randomized algorithm $\alg$ is \emph{$(\eps,\delta)$-differentially private (DP)} (with respect to the neighbor relation on the universe of the datasets) if for all events $S$ in the output space of $\alg$ and all  neighboring datasets $X$ and $X'$,
$$\Pr[\alg(X) \in S] \leq \exp(\eps)\cdot \Pr [\alg(X') \in S]\,+\delta.$$
When $\delta=0$, the algorithm is $\eps$-differentially private (sometimes also called ``purely differentially private'').
\end{defnt}

Differential privacy can be defined with respect to any notion of neighboring datasets. 
When datasets are represented as vectors, datasets $X$ and $Y$ are considered neighbors if they differ in one entry. 
In the context of graphs, there are two natural notions of neighboring graphs that can be used in the definition: edge-neighboring and node-neighboring. 
We use predominantly the former, but define both to make discussion of previous work clear.

\begin{defnt}\label{def:edge-neighboring}
Two graphs $G = (V, E)$ and $G' = (V', E')$ are \defn{edge-neighboring}
if $G$ and $G'$ differ in exactly one edge, that is,
if  $V = V'$ and $E$ and $E'$ differ in exactly one element.
Two graphs are \defn{node-neighboring} if one can be obtained from the other by removing a node and its adjacent edges.
\end{defnt}

If the datasets are graphs with edge (respectively, node) neighbor relationship, we call a differentially private algorithm simply {\em edge-private} (respectively, {\em node-private}).

\subsection{The local model} 
The definition of differential privacy implicitly assumes a trusted curator that has access to the data, runs a private algorithm on it, and releases the result. This setup is called the {\em central model} of differential privacy. In contrast, in the {\em local model} of differential privacy, each party participating in the computation holds its own data. The interaction between the parties is coordinated by an algorithm $\alg$ that accesses data via {\em local randomizers}. A local randomizer is a differentially private algorithm that runs on the data of one party. In the context of graph datasets, the input graph is distributed among the parties as follows: each party corresponds to a node of the graph and its data is the corresponding row in the adjacency matrix of the graph. 
 In each round of interaction, the algorithm $\alg$ assigns each party a local randomizer (or randomizers) that can depend on the information obtained in previous rounds.

We adapt the definition of local differential privacy from~\cite{kasiviswanathan2011can,JosephMNR19} to the graph setting.
Consider an undirected graph $G=([n],E)$ represented by an $n\times n$ adjacency matrix $\Adj$. Each party $i\in[n]$ holds the $i$-th row of $\Adj$, denoted $\adjb_{i*}$. We sometimes refer to $\adjb_{i*}$ as the {\em adjacency vector} of party $i$. Entries of $\Adj$ are denoted $a_{ij}$ for $i,j\in[n].$

\begin{defnt}[Local Randomizer]\label{def:local-randomizer}
    Let $\eps>0$ and $\delta\in[0,1)$.
    An \defn{$(\eps,\delta)$-local randomizer} $R: \{0,1\}^n \rightarrow \rangeout$ is an $(\eps,\delta)$-edge DP 
    algorithm that takes as input %
    the set of neighbors of one node, represented by an adjacency vector $\adjb\in\{0,1\}^n$.  In other words, $\Pr\left[R(\adjb) \in Y\right] \leq e^{\eps} \cdot
    \Pr\left[R(\adjb') \in Y\right]+\delta$ for all $\adjb$ and $\adjb'$
    that differ in one bit
    and all sets of outputs $Y \subseteq \rangeout$.
    The probability is taken over the random coins of $R$ (but \emph{not} over the choice of the input). When $\delta=0$, we say that $R$ is an $\eps$-local randomizer.
\end{defnt}

A randomized algorithm $\alg$ on a distributed graph is 
$(\eps, \delta)$-\ledp if it satisfies~\cref{def:ldp}.

\begin{defnt}[Local Edge Differential Privacy]\label{def:ldp}
A \defn{transcript} $\pi$ is a vector consisting of 5-tuples $(S^t_U, S^t_R, S^t_\eps, S^t_\delta, S^t_Y)$ -- encoding the set of parties chosen, set of randomizers assigned, set of randomizer privacy parameters, and set of randomized outputs produced -- for each round $t$. Let $S_\pi$ be the collection of all transcripts and $S_R$ be the collection of all randomizers. Let $\perp$ denote a special character indicating that the computation halts.
An \defn{algorithm} in this model is a function $\alg: S_\pi \to 
(2^{[n]} \times 2^{S_R} \times 2^{\mathbb{R}^{\geq 0}} \times 2^{\mathbb{R}^{\geq 0}})\; \cup \{\perp\}$
mapping transcripts to sets of parties, randomizers, and randomizer privacy parameters. The length of the transcript, as indexed by $t$, is its round complexity.

Given $\eps\geq 0$ and $\delta\in[0,1)$,  a randomized algorithm $\alg$ on a (distributed) graph $G$ is \defn{$(\eps,\delta)$-locally edge differentially private (LEDP)} if the algorithm that outputs the entire transcript generated by $\alg$ is $(\eps,\delta)$-edge differentially private on graph $G.$
 When $\delta=0$, we say that $\alg$ is an $\eps$-LEDP.
 
 If $t=1$, that is, if there is only one round, then $\alg$ is called \defn{noninteractive}. Otherwise, $\alg$ is called \defn{interactive}.
\end{defnt}

Observe that a noninteractive LEDP algorithm is specified by a local randomizer for each node and a postprocessing algorithm $\cP$ that 
takes the outputs of the local randomizers as input.

We 
use a local algorithm
known as \emph{randomized response}, initially due to~\cite{Warner65},  but since adapted to differential privacy~\cite{kasiviswanathan2011can}.

\begin{defnt}[Randomized Response]\label{def:randomized-response}
Given a privacy parameter $\eps>0$ and a $k$-bit vector $\adjb$,
the algorithm $\rr_\eps(\adjb)$ outputs a $k$-bit vector,
where for each $i\in[k],$ bit $i$ is $a_i$ with probability $\frac {e^\eps}{e^\eps+1}$ and $1-\adj_i$ otherwise.
\end{defnt}

\begin{thm}[Randomized Response is $\eps$-LR]\label{thm:rr-ldp}
Randomized Response is an $\eps$-local randomizer. %
\end{thm}
\ificalp
\else
 Additional privacy tools can be found in \cref{sec:append-privacy-tools}.
\fi

\section{The Noninteractive Lower Bound}\label{sec:one-round-lower}

In this section, we prove 
\cref{thm:noninteractiv-lb-intro}, which we restate formally here.

\begin{thm}\label{thm:noninteractiv-lb}
There exists a family of graphs, such that every noninteractive $(\eps,\delta)$-LEDP algorithm with $\eps\in(0,\frac{1}{20})$ and $\delta\in[0,\frac {1}{100})$ 
that gets an $n$-node graph from the family as input and approximates the number of triangles in the graph within additive error %
$\alpha$ with probability at least $1-\frac{1}{3^6\cdot 2^{7}}$, must have $\alpha=\Omega(n^2)$.
\end{thm}

At a high level, the lower bound is proved by showing that a noninteractive local algorithm for counting triangles can be  used  to mount a reconstruction attack on a secret dataset $X$ in the central model of differential privacy.
A groundbreaking result of Dinur and Nissim~\cite{DiNi03}---generalized in subsequent works \cite{DworkMT07,DworkY08,KasiviswanathanRSU10,KasiviswanathanRS13,anindya_lb}---shows that if an algorithm answers too many random linear queries on a sensitive dataset of $N$ bits  too accurately then a large constant fraction of the dataset can be reconstructed.
This is referred to as a    ``reconstruction attack''.
Specifically,  Dinur and Nissim show that $N$ random linear queries answered to within $\pm O(\sqrt{N})$ are sufficient for reconstruction. 
It is well known that if the output of an algorithm on a secret dataset can be used for reconstruction, then this algorithm is not differentially private. 
This line of reasoning leads to a lower bound of $\Omega(\sqrt{N})$ on the additive error of any differentially private algorithm answering $N$ random linear queries.

Suppose  we could show that an LEDP triangle counting algorithm with $O(n^2)$ additive error can be used to construct a DP algorithm for answering $N$ linear queries with $O(\sqrt N)$ additive error on some dataset of size $N$ --- then by the above, we reach a contradiction to the privacy of the algorithm.  
While indeed a  triangle counting algorithm can be used to answer a \emph{single} linear query,  the main challenge is that the Dinur-Nissim reconstruction attack requires answering not one, but rather $n$, linear queries on the same dataset. 
Let $\alg$ be an $(\eps,\delta)$-LEDP triangle counting algorithm.
If we naively try to answer each linear query to $X$ using a new invocation of the triangle counting algorithm in a black-box manner, this would result in $n$ invocations of $\alg$. This in turn would cause the privacy parameters to grow linearly with $n$, making the privacy breach vacuous. That is, the result would be of the following sort. An $(\eps, \delta)$-LEDP algorithm for triangle counting with low additive error implies an $(O(\eps n), O(n\delta))$-DP algorithm for answering linear queries with low additive error. Since the latter statement is too weak to be used with the results of Dinur and Nissim, we  take a different approach. 

To avoid making $n$ invocations of a triangle counting algorithm, we  develop  a new type of reconstruction attack on a secret dataset $X$, where the set of allowed linear queries has  a special combinatorial structure. We call the new type of queries  \emph{outer-product queries}. We  show that, given access to an $(\eps, \delta)$-LEDP algorithm $\alg$ that approximates the number of triangles up to $O(n^2)$ additive error, we can 
 design a $(2\eps, 2\delta)$-DP algorithm $\cB$ for answering $\Theta(n^2)$ outer-product queries on dataset $X$ of size $N=n^2$, so that a constant fraction of them is answered with $O(n)$ additive error. 
 (The dataset size is $n^2$, so asymptotically the number of random queries and the required accuracy are the same as in the Dinur-Nissim attack.) 
 The main insight is that instead of using $\alg$ as a black-box, we use it in a ``gray-box'' manner. This allows us to answer all $\Theta(n^2)$ queries without degrading the privacy parameters of $\cB$.
  This  in turn allows us to reconstruct $X$, which is a contradiction to the privacy of  algorithm $\cB$, and thus also to the privacy of algorithm $\alg$. Hence, we  conclude that any LEDP triangle-counting algorithm must have $\Omega(n^2)$ additive error.

The rest of  \cref{sec:noninteractive-reduction} is organized as follows. 
In \cref{sec:noninteractive-reduction}, we define  outer-product queries and show that an $(\eps, \delta)$-DP algorithm $\alg$ for triangle-counting with low additive error can be used to construct a $(2\eps, 2\delta)$-DP algorithm $\cB$ for answering outer-product queries with low additive error. 
In \cref{sec:anti-concentration},   we prove an anti-concentration result for random outer-product queries. 
In \cref{sec:attacks}, we use the anti-concentration result to show that  an algorithm $\cB$ that accurately  answers $\Theta(n^2)$ outer-product queries on a sensitive dataset $X\in \{0,1\}^{n\times n}$ can be used to reconstruct most of $X$ and complete the proof of \cref{thm:noninteractiv-lb}.

\subsection{Reduction from Outer-product Queries to Triangle Counting}\label{sec:noninteractive-reduction}

In this section, we prove \cref{lem:answering-queries}, which is at the heart of our reduction. It
shows that, given access to an $(\eps,\delta)$-LEDP algorithm $\alg$ for approximating the number of triangles with low additive error, we can construct an 
$(2\eps,2\delta)$-DP algorithm $\cB$ (in the central model) that accurately answers $\Theta(n^2)$ outer-product queries on a sensitive dataset $X$.
We start by formally defining this new class of queries.

\begin{defnt}[Outer-product queries]\label{def:outerproduct-queries} 
Let $X\in\{0,1\}^{n\times n}$. An {\em outer-product query} to $X$ specifies two vectors $A$ and $B$ of length $n$ with entries in $\{-1,1\}$ and returns $A^TXB,$ that is, $\sum_{i,j\in[n]} A_iX_{ij} B_j$.
\end{defnt}

 Let $\gamma$ be the desired reconstruction parameter that indicates that the attack has been successful if we reconstruct at least $(1-\gamma) n^2$ bits of $X$ correctly. 
 (Later, in \cref{sec:attacks}, $\gamma$ will be set to $\frac{1}{9}$ and the number of queries, $k$, will be set to $\Theta(n^2)$.)
 
\begin{lem}[Answering Outer-product Queries via Triangle Counting]\label{lem:answering-queries}
Let $\eps,\delta>0$ and $\gamma\in(0,1/2)$.
Assume that there is a noninteractive $(\eps,\delta)$-LEDP algorithm $\alg$ that, for every $3n$-node graph, approximates the number of triangles  with probability at least $1-\frac{\gamma^2}{9\cdot 128}$  and has additive error  at most {$\frac{\sqrt \gamma n^2}{20}$}. Then there is an $(2\eps,2\delta)$-DP algorithm $\cB$ {in the central model}  that, for every secret dataset $X\in\{0,1\}^{n\times n}$ and every set of $k$ outer-product queries $(A^{(1)},B^{(1)}),\dots,(A^{(k)},B^{(k)})$,
gives answers $a_1,\dots, a_k$ satisfying
\begin{align}\label{eq:outer-product-queries-accuracy}
\Pr\left[\left|\left\{\ell\in[k]: 
\Big|(A^{(\ell)})^T X B^{(\ell)}-a_\ell\Big|>\frac{\sqrt{\gamma}n}{{4}}\right\}\right|>\frac{\gamma^2 k}{{64}}\right]\leq \frac 1 6.
\end{align}
That is, with probability at most $5/6$, for every dataset $X$ and a set of $k$ outer-product queries, Algorithm $\cB$ answers inaccurately  at most   $\frac{\gamma^2 k}{64}$ of the $k$ queries, where ``inaccurately'' means with additive error more than $\frac{\sqrt \gamma n}{4}$.
\end{lem}
\begin{proof}
Consider an algorithm $\alg$ described in the premise of the lemma.
Since $\alg$ is local noninteractive, it is specified by a local randomizer $R_v(\adjb)$ for each vertex $v$, as well as a postprocessing algorithm $\cP$. Each randomizer takes an adjacency vector $\adjb\in\{0,1\}^n$ as input and passes its output to $\cP.$  
Next, we define algorithm $\cB$ that, given a sensitive dataset $X$ and a set of $k$ outer-product queries, uses the randomizers and the postprocessing algorithm as subroutines to obtain accurate answers to the outer-product queries.

Fix a dataset $X\in\{0,1\}^{n\times n}$. For each outer-product query $(A,B)$, algorithm $\cB$ constructs several
corresponding {\em query} graphs. All query graphs are on the same vertex set $V$ of size ${3n}$, partitioned into 
{three sets $U_1,U_2,$ and $W$}
of size $n$. The vertices in $U_t$ for $t\in\{1,2\}$ are denoted $u_{t1},\dots,u_{tn}$. %
{The vertices of $W$ are denoted $w_1,\dots, w_n$.}
See \Cref{fig:non-interactive-lb} for an illustration.

Algorithm $\cB$  first forms a bipartite {graph $G_X$} %
with parts $U_1$ and $U_2$ with $X$ as the adjacency matrix; that is, it adds an edge $(u_{1i},u_{2j})$ for each $i,j\in[n]$ with $X_{ij}=1.$ {We call} $G_X$ the \emph{secret subgraph}, {because it will be included as a subgraph in every query graph and it will be the only part of that graph that contains any information about the original sensitive dataset $X$}. Note that {$G_X$}  %
does not depend on the outer-product query.
 The remaining edges of each query subgraph are  between $U_1\cup U_2$ and %
 {$W$}
 and are specific to each query graph, so that overall  the resulting graph is
 tripartite. 
For each $v\in U_1\cup U_2$, let $\Gamma_X(v)$ denote the neighbors of $v$ in the secret subgraph $G_X$. A key idea in the construction is that every node in the secret subgraph $G_X$ will have one of \talya{only} two possible neighborhoods in each query graph. This allows algorithm $\cB$ to simulate triangle-counting computations on all query graphs by invoking a local randomizer on each vertex in $U_1\cup U_2$ only twice.
{For each vertex $v\in U_1\cup U_2,$
algorithm $\cal B$ runs its local randomizer $R_v(\cdot)$ twice:
once with the adjacency list specified by  
$\Gamma_X(v)$ 
and once  with the adjacency list specified by $\Gamma_X(v)\cup W$. 
Algorithm $\cal B$ then records the output of the former  invocation as
as $r_0(v)$,  and the latter as $r_1(v)$.}

\begin{figure}
    \begin{center}
\includegraphics[width=0.4\textwidth]{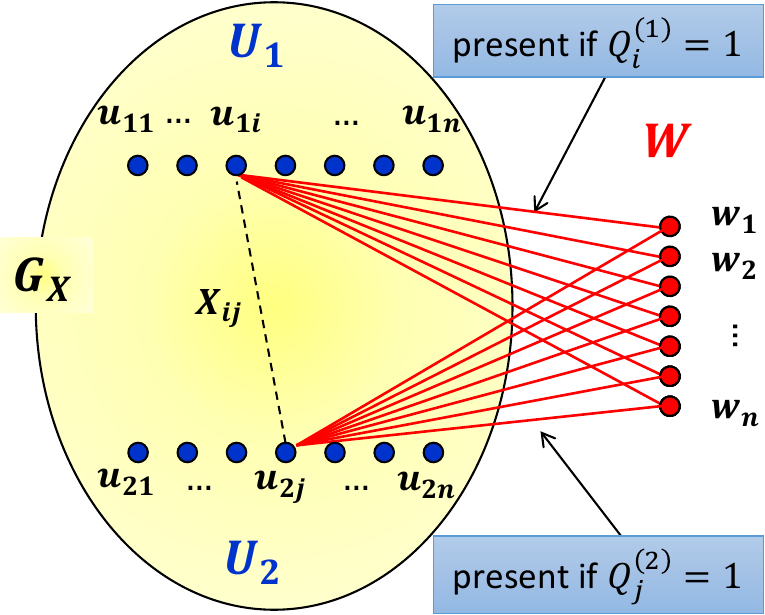}
    \end{center}
    \caption{The construction of the query graph $G_{X,Q}$. 
    Each of the parts $U_1, U_2, W$ consists of $n$ nodes. 
    The dashed line is an edge iff $X_{ij}=1$.
    Only the subgraph $G_X$ (induced by  $U_1\cup U_2$) holds secret information. 
    }
    \label{fig:non-interactive-lb}
\end{figure}

\ificalp
\else
By \cref{thm:composition}, 
\fi
the algorithm that simply outputs the vector of all  $4n$ responses of the local randomizers is $(2\eps,2\delta)$-DP by composition, because each bit of $X$ is encoded as a potential edge and used in two executions of the randomizers for its endpoints, where each execution (of all randomizers) is $(\eps,\delta)$-LEDP. In the remaining steps, algorithm $\cB$  only postprocesses the vector of responses, and thus it is  $(2\eps,2\delta)$-DP.  

Next, we describe how to postprocess the vector of responses to obtain an answer to an outer-product query $(A,B).$ To answer each outer-product query, algorithm $\cB$ will first obtain answers to three linear queries that we call {\em submatrix queries}. Submatrix queries are defined the same way as outer-product queries, except that vectors $A$ and $B$ have entries in $\{0,1\}$ instead of $\{-1,1\}$. Next, we explain how to answer submatrix  queries on $X$, deferring to \cref{claim:query-simulation} the description of the simulation of each outer-product query with submatrix queries.

To answer a submatrix query $Q=(Q^{(1)},Q^{(2)})$ on dataset $X$, algorithm $\cB$ completes the secret subgraph $G_X$ to a query graph $G_{X,Q}$ as follows.
For each vertex  $u_{ti}\in U_1\cup U_2,$ where $t\in\{1,2\}$ and $i\in[n]$, {it adds  edges determined by $Q^{(t)}$: specifically, if $Q^{(t)}_i=1$, it adds edges from $u_{ti}$ to all vertices in $W$.}
The next claim states the relationship between the number of triangles in $G_{X,Q}$ and the answer to the submatrix query $Q$.

\begin{clm}\label{claim:submatrix-query-answer}
The number of triangles in graph $G_{X,Q}$ is equal to $n\cdot (Q^{(1)})^T X Q^{(2)}.$ 
\end{clm}
\begin{proof}
Observe that $G_{X,Q}$ is tripartite with parts $(U_1,U_2,{W})$,
so all triangles must have one vertex in each part. %
The answer to the submatrix query $Q=(Q^{(1)},Q^{(2)})$ is
$$(Q^{(1)})^T X Q^{(2)}=\sum_{i,j\in[n]} Q^{(1)}_i Q^{(2)}_jX_{ij}.$$
For each term in the sum, both $u_{1i}$ and $u_{2j}$ are adjacent to all nodes in {$W$}
iff $Q^{(1)}_i =Q^{(2)}_j=1.$ 
If the edge $(u_{1i},u_{2j})$ is  present in the graph, then this results in $n$ triangles.
 Thus, each term where $Q^{(1)}_i =Q^{(2)}_j=X_{ij}=1$ corresponds to $n$ triangles of the form $(u_{1i},u_{2j},
 {w_\ell}
 )$, where $\ell\in[n]$. All other terms create no 
 triangles, since either  $X_{ij}=0$, in which case the edge $(u_{1i},u_{2j})$ is not present in the graph, or either  $Q^{(1)}_i=0$ or $Q^{(2)}_j=0,$ in which case $u_{1i}$ and $u_{2j}$ do not have common neighbors.
\end{proof}

To answer a submatrix query $Q$, algorithm $\cB$ simulates a call to the triangle-counting algorithm $\alg$ on the corresponding query graph $G_{X,Q}$. 
First, $\cB$ runs the local randomizers for the vertices in {$W$}
with their adjacency vectors specified by the graph $G_{X,Q}.$ Note that these vertices do not have access to any private information, so this operation does not affect privacy.
For each %
vertex $u_{ti}\in U_1\cup U_2$, where $t\in\{1,2\}$ and $i\in[n]$, algorithm $\cB$ uses the result {$r_b(u_{ti})$} from the previously run randomizer, where $b=Q^{(t)}_i$; e.g., if $Q^{(1)}_i=0$, then $\cB$ uses the result $r_0(u_{1i})$, and if $Q^{(1)}_i=1$, it uses the result $r_1(u_{1i})$.
Now algorithm $\cB$ has results from all vertex randomizers on the graph $G_{X,Q}$ and it simply  runs the postprocessing algorithm $\cP$ on these results. To obtain the answer to the submatrix query, $\cB$ divides the  output of $\cP$ by $n$.

Finally, algorithm $\cB$ answers each outer-product query as specified in the following claim, by getting answers to {three} submatrix queries.

\begin{clm}\label{claim:query-simulation}
An outer-product query to $X$ can be simulated with {three} submatrix queries to $X$. Moreover, if {all three} submatrix queries are answered with additive error at most $\alpha$, then the outer product query can be answered with additive error at most $5\alpha.$
\end{clm}
\begin{proof}
Consider an outer-product query to an $n\times n$ matrix $X$ specified by $A,B\in\{-1,1\}^n.$ Define $n$-bit vectors $A'=\frac 12 (A +\vec{1})$ and $A''=\frac 12 (-A +\vec{1}),$ where $\vec{1}$ denotes a vector of 1s of length $n.$ Define $B'$ and $B''$ analogously. %
Then, as illustrated in \Cref{fig:submatrix_query},
$$A^T X B= 2((A')^T X B'+ (A'')^T X B'')- {\vec{1}^T X\vec{1}}.$$
\ifnum\drawFigTwo=1
\fi
That is, the answer to the outer-product query $(A,B)$ can be computed from the answers to the submatrix queries $(A',B'),(A'',B''),$ {and $(\vec{1},\vec{1})$,} and the additive error increases from $\alpha$ to $5\alpha$, as stated.
\end{proof}

\ifnum\drawFigTwo=1
\begin{figure}
    \includegraphics[width=\textwidth]{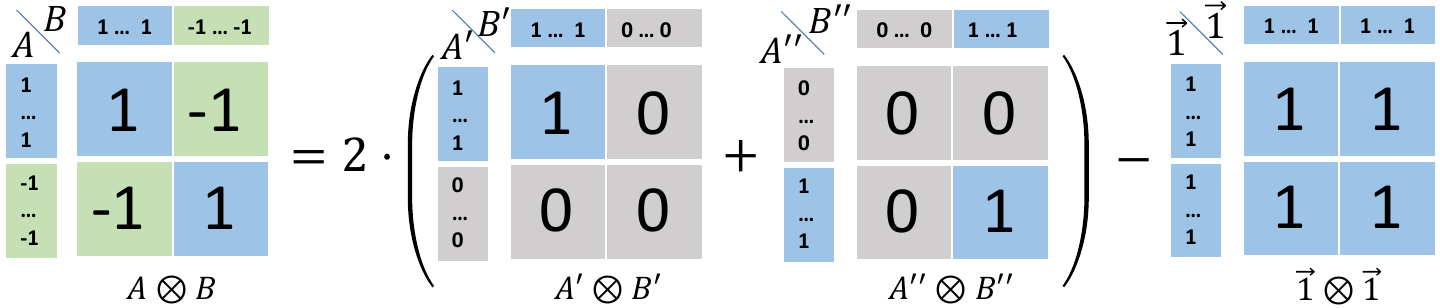}
    \caption{
    Every outer-product query can be simulated using three submatrix queries. For the illustration, the entries of all vectors are rearranged to group the same values together. The outer product is denoted $\otimes.$
}
    \label{fig:submatrix_query}
\end{figure}
\else
\fi

{It remains to prove that algorithm $\cB$ satisfies \eqref{eq:outer-product-queries-accuracy}.}
By the assumption on $\alg$, for every graph $G$, algorithm $\alg$ returns the number of triangles in $G$ 
within an additive error at most {$\frac{\sqrt \gamma n^2}{20}$}
with probability at least $1-\frac{\gamma^2}{{9\cdot 128}}$.
Given a secret dataset $X$ 
and $k$ outer-product queries, algorithm $\cB$ first creates $k$ triples of submatrix queries corresponding to the outer-product queries.  Then, $\cB$ uses $\alg$ as a gray box to answer  all $3k$  submatrix queries  simultaneously. Recall that this is achieved  by  invoking  the local randomizers on vertices holding private information (that is, vertices in parts $U_1, U_2$) twice, once for each potential value of the bit that  corresponds to this vertex in a specific query. Then for each individual submatrix query,  one
local randomizer is invoked on each of the
{$n$ vertices in $W$}
with the adjacency list that corresponds to that specific query graph. Then, to answer each specific submatrix query, algorithm $\cB$ combines the new outputs of the vertices {from $W$}
with the stored outputs from running randomizers on  $U_1\cup U_2$ that correspond to that specific query, and invokes the postprocessing algorithm $\cP$ on this vector of ${3n}$ outputs. Finally, $\cB$ divides $\cP$'s answer by $n$ to obtain the answer to the submatrix query. 

Each invocation of $\cP$ by $\cB$ simulates one triangle-counting computation. Overall, we have $3k$ (\emph{dependent}) simulated triangle-counting computations.
By the assumption on $\alg$, stated  in the premise of
\cref{lem:answering-queries}, the postprocessing algorithm  $\cP$ answers each simulated triangle-counting computation inaccurately (i.e., with additive error {exceeding} {$\frac{\sqrt \gamma n^2}{20}$}) with probability at most $\frac{\gamma^2}{{9\cdot 128}}$ (where this probability is taken over the random coins of the individual ${3n}$ local randomizers, as well as the random coins of $\cP$).
Overall, there are $3k$ (dependent) simulations, and so the expected number of simulated triangle-counting computations for which $\alg$  returns additive error greater than {$\frac{\sqrt \gamma n^2}{20}$} is {at most}
$\frac{\gamma^2\cdot (3k)}{{9\cdot 128}}=\frac{\gamma^2\cdot k}{6\cdot {64}}$.
Hence, 
by Markov's inequality, the probability that the number of inaccurate simulated triangles queries exceeds $\frac{\gamma^2\cdot k}{{64}}$ is at most $\frac 1 6$.

Condition on  the event that at most $\frac{\gamma^2\cdot k}{64}$ of the  triangle-counting computations  are answered inaccurately, so that the remaining computations are answered with error at most $\alpha=\frac{\sqrt \gamma n^2}{20}$, and denote this even by $E$. 
Recall that each triangle-counting computation is used to answer a single submatrix query, and that  by \cref{claim:submatrix-query-answer}, if a triangle-counting computation is answered with additive error $\alpha$, then the corresponding submatrix query is answered with additive error $\alpha/n$.  
Hence, by the above conditioning, at most $\frac{\gamma^2\cdot k}{64}$ of the submatrix queries are answered with additive error greater than $\alpha/n$. 
Each inaccurate answer to a triangle-counting computation can spoil the answer to at most one outer-product query. 
Furthermore, by \cref{claim:query-simulation}, if all three submatrix queries used to compute  a single outer-product query are answered to within additive error $\alpha/n$, then the outer-product query is answered to within  additive error $5\alpha/n$. Hence, by the above conditioning, at most $\frac{\gamma^2 \cdot k}{64}$ of the outer-product queries are answered with additive error greater than $5\alpha/n=\frac{\sqrt \gamma n}{4}$.
Since event $E$ occurs with probability at least $5/6$, 
we get that with probability at least $5/6$, the fraction of outer-product queries that are answered with additive error greater than $\frac{\sqrt \gamma n}{4}$ is at most $\frac{\gamma^2 \cdot k}{64}$.
This completes the proof of \cref{lem:answering-queries}.
\end{proof}

\subsection{Anti-Concentration for Random Outer-Product Queries}\label{sec:anti-concentration}
\newcommand{\RV}{{U}}
In this section, we prove \cref{lem:anti-concentration}.
To analyze our reconstruction attack, we will consider the differences between the true dataset $X$ and a potential reconstructed dataset $Y$.
Let $M$ denote $X-Y$. Then, for an outer-product query $(A,B)$, the difference between the answers to this query on dataset $X$ and on dataset $Y$ is $A^T X B- A^T Y B= A^T M B.$ The main result of this section shows that if $X$ and $Y$ differ on many entries (that is, $M$ has lots of nonzero entries) then a random outer-product query is likely to produce significantly different answers on $X$ and $Y$.

\begin{proof}[Proof of \cref{lem:anti-concentration}]
Let $Z_{ij}=A_iB_j$ for all $i,j\in[n],$ and $\RV=A^T M B.$ 
We prove the lemma by computing the expectation and the second and the fourth moments of $\RV$, and then apply the Paley-Zigmund inequality 
\ificalp\else
(stated in~\cref{prel:paley-zygmund}) 
\fi
to $\RV^2.$

By independence of $A_i$ and $B_j$  for all $i,j\in[n],$ we have $\E[Z_{ij}]=\E[A_i]\cdot\E[B_j]=0$ and
$\V[Z_{ij}]=\E[Z_{ij}^2]=\E[A_{i}^2B_j^2]=1$.
By definition of $\RV$ and the linearity of expectation, 
$$\E[\RV]
=\E[A^T M B]
=\E\Big[\sum_{i,j\in[n]} M_{i,j} Z_{i,j}\Big]
=\sum_{i,j\in[n]} M_{i,j} \E[Z_{i,j}]
=0.$$

Note that random variables $Z_{ij}$ are pairwise independent. This is an important feature of random outer-product queries and the main reason to use them instead of the submatrix queries. This feature greatly simplifies the analysis. Since $\RV$ is unbiased, $\E[\RV^2]=\V[\RV]$. By pairwise independence of $Z_{ij}$,
$$\V[\RV]
=\V\Big[\sum_{i,j\in[n]}M_{ij}Z_{ij}\Big]
=\sum_{i,j\in[n]}M_{ij}^2\ \V[Z_{ij}]
=\sum_{i,j\in[n]}M_{ij}^2
=m.
$$

Next, we give an upper bound on the 4th moment of $\RV$. 
\begin{clm}\label{claim:4th-moment}
$\E[\RV^4]\leq 9n^4.$
\end{clm}
\begin{proof}
We use the definition of $\RV$, write it out as a sum, and multiply out the terms of the product:
\begin{align}
\E[\RV^4]
&=\E[(A^T M B)^4]
=\E\Big[\Big(\sum_{i,j\in[n]}M_{ij}Z_{ij}\Big)^4\Big]\nonumber\\
&=\sum_{(i_1,j_1),\dots,(i_4,j_4)\in[d]\times[d]} M_{i_1j_1}M_{i_2j_2}M_{i_3j_3}M_{i_4j_4}
       \  \E[Z_{i_1j_1}Z_{i_2j_2}Z_{i_3j_3}Z_{i_4j_4}],\label{eq:4-term-expectation}
       \end{align}
where \eqref{eq:4-term-expectation} is obtained by using the linearity of expectation. Next, we evaluate the expectation of the product in \eqref{eq:4-term-expectation}:
\begin{align*}
\E[Z_{i_1j_1}Z_{i_2j_2}Z_{i_3j_3}Z_{i_4j_4}]
&=\E[A_{i_1}B_{j_1}A_{i_2}B_{j_2}A_{i_3}B_{j_3}A_{i_4}B_{j_4}]\nonumber\\
&=  \E[A_{i_1}A_{i_2}A_{i_3}A_{i_4}] \ \E[B_{j_1}B_{j_2}B_{j_3}B_{j_4}],
       \end{align*}
where the last equality follows by independence of $A_i$ and $B_j$ for all $i,j\in[n].$
The expression $\E[A_{i_1}A_{i_2}A_{i_3}A_{i_4}]$ is 0 if
at least one of the indices appears only once in the tuple $(i_1,i_2,i_3,i_4)$,  since, in this case, we can use the independence of the corresponding factor $A_i$ from the remaining factors to represent this expression as $\E[A_i]$ multiplied by the expectation of the product of the remaining factors. Since $\E[A_i]=0$ for all $i\in[n]$, the overall expression evaluates to 0. 

Note that if one of the factors appears exactly three times, then another factor appears exactly once. Therefore, the remaining case is when each factor appears an even number of times. If there are two factors, say $A_i$ and $A_j$ that appear twice, then the expression evaluates to $\E[A_i^2A_j^2]=1$. It also evaluates to 1 when $i=j$. 

Thus, each term in \eqref{eq:4-term-expectation} is either 0 or 1. By symmetry, it can potentially be 1 only if each index in the tuple $(i_1,i_2,i_3,i_4)$ and each index in the tuple $(j_1,j_2,j_3,j_4)$ appears an even number of times.
It remains to give an upper bound on the number of such terms.
There are $n\choose 2$ ways to choose two distinct $i$-indices and ${4\choose 2}=6$ possible positions for them in the 4-tuple. In addition, there are $n$ ways to choose an index that appears 4 times in the 4-tuple. So, the number of possibilities for nonzero $\E[A_{i_1}A_{i_2}A_{i_3}A_{i_4}]$ is at most $3n^2.$ The same bounds holds for $\E[B_{j_1}B_{j_2}B_{j_3}B_{j_4}]$. Consequently, the number of terms equal to 1 in \eqref{eq:4-term-expectation} is at most $9n^4$. Thus, the sum evaluates to at most $9n^4$. This completes the proof of~\cref{claim:4th-moment}.
\end{proof}

Since $\RV^2$ is a nonnegative random variable with finite variance,  the Paley-Zygmund inequality gives that, for all $\theta\in[0,1]$,
$$\Pr\big[\RV^2>\theta\ \E[\RV^2]\big]
\geq (1-\theta)^2 \frac{(\E[\RV^2])^2}{\E[\RV^4]}
\geq (1-\theta)^2 \frac{m^2}{9n^4}
\geq (1-\theta)^2 \frac{(\gamma n^2)^2}{9n^4}
= (1-\theta)^2 \frac{\gamma^2}{9},
$$
where the last inequality uses the bound $m\geq \gamma n^2$ stated in the lemma. Finally,
we set $\theta=1/4$ and get:
\begin{align*}
\Pr\Big[|A^T M B|>\frac{\sqrt{m}}2\Big]
=\Pr\Big[|\RV|>\frac{\sqrt{m}}2\Big]
=\Pr\Big[\RV^2>\frac{m}4\Big]
\geq \frac{3^2}{4^2} \frac{\gamma^2}9
=\frac {\gamma^2}{16},
\end{align*}
completing the proof of~\cref{lem:anti-concentration}.
\end{proof}

\subsection{Reconstruction Attack Using Outer-Product Queries}\label{sec:attacks}
To simplify notation in this section, we represent our datasets and outer-product queries as vectors. Formally, $X$ here denotes the vectorization of the original sensitive dataset, i.e., a vector in $\{0,1\}^{n^2}.$ For an outer-product query $(A,B),$ we let $Q\in\{0,1\}^{n^2}$ represent the vectorization of $A\otimes B$, the outer product of $A$ and~$B$. (In other words, $Q$ is the Kronecker product of $A$ and~$B$.) Then the answer to the query is the dot product $Q\cdot X.$

In this section, we define and analyze the attacker's algorithm $\cC$ and complete the proof of \cref{thm:noninteractiv-lb}. The attacker $\cC$ runs algorithm $\cB$ from~\cref{sec:noninteractive-reduction} on the sensitive dataset $X$ and a set of $k$ random outer-product queries $Q_1,\dots, Q_k$ to obtain answers $a_1,\dots,a_k.$ For all $\ell\in[k]$, we call the answer $a_\ell$ {\em accurate for a dataset $Y$} if $|Q_\ell\cdot Y -a_\ell|\leq {\frac{\sqrt{\gamma}n}{4}}$; otherwise, we call $a_\ell$ {\em inaccurate for $Y$}. The attacker $\cC$ outputs any dataset $Y^*\in\{0,1\}^{n^2}$ for which at most {$\frac{\gamma^2 k}{64}$}  answers among $a_1,\dots,a_k$ are inaccurate for $Y^*.$ By \cref{lem:answering-queries}, the probability that $X$ satisfies this requirement is at least $\frac 5 6$. If this event occurs, algorithm~$\cC$ will be able to output some $Y^*.$ (Otherwise, the attack fails.)

{Next, we analyze the attack. Let $\|X-Y\|_1$ denote the Hamming distance between datasets $X$ and $Y$. Call a dataset $Y$ {\em bad} if $\|X-Y\|_1 > \gamma n^2$, i.e., if it differs from $X$ on more than $\gamma n^2$ entries. We will show that $\cC$  is unlikely to choose a bad dataset as $Y^*.$}
Fix a bad dataset $Y$.
Let $M=X-Y,$ and observe that $M$ has $m>\gamma n^2$ nonzero entries.
We say that a set of queries $\{Q_1,\dots, Q_k\}$ {\em catches} the  dataset $Y$ if more than {$\frac {\gamma^2 k}{32}$} entries in $(|Q_1\cdot M|,\dots, |Q_k\cdot M|)$ exceed {$\frac{\sqrt{\gamma}n}{2}$}. 
\begin{lem}\label{lemma:failure-to-catch}
Suppose the attacker $\cC$ makes $k=\frac{128n^2}{\gamma^2}$ uniformly random outer-product queries.
Then the probability that there exists a bad dataset not caught by the attacker's set of queries 
is at most $\frac 1 6$.
\end{lem}
\begin{proof}
Consider a set of $k$  uniformly random outer-product queries $\{Q_\ell\}_{\ell\in[k]}$. Fix a bad dataset $Y$. Then {$\|X-Y\|_1> \gamma n^2$.}
Let $M=X-Y.$

 For every $\ell\in[k]$, let $\chi_\ell=1$ if $\left\lvert Q_\ell\cdot M\right\rvert>{\frac{\sqrt{\gamma}n}{2}}$, and otherwise let $\chi_\ell=0$. Also, let $\chi=\sum_{\ell=1}^k \chi_\ell.$
By definition, 
the difference vector $M=X-Y$ has more than  $\gamma n^2$ nonzero entries. By the anti-concentration bound in \cref{lem:anti-concentration}, $\Pr\Big[|Q_\ell\cdot M|>\frac{\sqrt \gamma n}{2}\Big]\geq \frac{\gamma^2}{16}$.
Therefore, $\expect[\chi_\ell]\geq \frac{\gamma^2}{16}$. 
\ificalp
\else
By the Chernoff bound 
(stated in \cref{prel:chernoff}),  
\fi
we have that for $k=\frac{128n^2}{\gamma^2}$ and for $n\geq 3$,
\[
\Pr\Big[ \chi\leq \frac{\gamma^2\cdot k}{32}\Big]\leq \exp\Big(-\frac{\gamma^2k}{128}\Big)=\exp\Big(-n^2 \Big)< \frac{1}{6\cdot 2^{n^2}}. 
\]
Hence, the set $\{Q_\ell\}_{\ell\in[k]}$ fails to catch each specific bad dataset  with probability at most $\frac{1}{6\cdot 2^{n^2}}$. By  a union bound over at most $2^{n^2}$ bad datasets, the probability that there exists a bad dataset not caught by the attacker's queries is at most 1/6.
\end{proof} 

\begin{lem}[Reconstruction Lemma]\label{lem:reconstruct} 
If algorithm $\cB$  has additive error at most {$\frac{\sqrt{\gamma} n}{4}$}  on all but at most {$\frac{\gamma^2 k}{64}$} answers, and the set of queries it uses catches all bad datasets $Y$, then the reconstruction attack is successful, that is, the attacker  $\cC$ outputs $Y^*$ that differs from $X$ on at most $\gamma n^2$ entries, {i.e.,  $\|X-Y^*\|_1\leq \gamma n^2$}.
\end{lem}
\begin{proof}
By the first premise of the lemma, the dataset $X$ ``disagrees" with at most {$\frac{\gamma^2 k}{64}$} of the answers~$a_\ell$. Hence, necessarily, the attacker $\cC$ outputs some dataset $Y^*$.
 Assume towards a contradiction that $Y^*$ is a bad dataset.
Let $\{Q_\ell\}_{\ell\in[k]}$ be the set of queries chosen by $\cB$. Let $M^*=X-Y^*$ be the difference vector.
By the triangle inequality,   
$|Q_\ell M^*|=|Q_\ell X-Q_\ell Y^*|
\leq |Q_\ell X -a_\ell|+|Q_\ell Y^*-a_\ell|$. From the first assumption in the lemma,   $|Q_\ell X-a_\ell|\leq \frac{\sqrt \gamma n}{4}$ 
for all but at most $\frac{\gamma^2 k}{64}$ of the queries.
By the description of the attack~$\cC$,  the output $Y^*$ is such that for all but  at most  {$\frac{\gamma^2 k}{64}$} of the queries, $|Q_\ell Y^*-a_\ell|\leq {\frac{\sqrt \gamma n}{4}}$. Therefore, for all but at most $\frac{\gamma^2 k}{32}$ of the queries, 
$|Q_\ell M^*|\leq |Q_\ell X -a_\ell|+|Q_\ell Y^*-a_\ell|\leq  \frac{\sqrt{\gamma} n}2$. 
 Since $\{Q_\ell\}_{\ell\in[k]}$ catches all bad datasets, it in particular catches $Y^*$, because  $Y^*$ is bad. By definition of catching, $|Q_\ell M^*|> \frac{\sqrt{\gamma} n}2$ for more than $\frac{\gamma^2 k}{32}$ of the values $Q_\ell M^*$. Hence, we have  reached a contradiction, implying that $Y^*$ is a good dataset. 
\end{proof}

The final ingredient for proving \cref{thm:noninteractiv-lb} is the following lemma, which is based on an argument of~\cite{anindya_lb}. Any algorithm that outputs a large fraction of its secret dataset is definitely not private, for any reasonable notion of privacy. 
\cref{lem:anindya-exp}
 states that such an algorithm is not differentially private.

\begin{lem}\label{lem:anindya-exp}
Let $\cC$ be an algorithm that takes as input a secret dataset $X$ in $\{0,1\}^N$ and outputs a vector in the same set, $\{0,1\}^N$. If $\cC$ is $(\eps,\delta)$-differentially private and $X$ is uniformly distributed in $\{0,1\}^N$, then 
$$\E\left[\|\cC(X)-X\|_1\right] \geq e^{-\eps} \paren{\tfrac 1 2 - \delta} N  \, .$$
\end{lem}

\cref{lem:anindya-exp} above only bounds the expectation of $\|\cC(X)-X\|_1$. The more sophisticated argument in \cite{anindya_lb} yields much tighter concentration results. We use the simpler version here since it allows for a self-contained presentation.

\begin{proof} \newcommand{\subst}[3]{{#1}_{{#2} \to {#3}}}
    Fix an index $i \in [N]$ and a bit $r\in\{0,1\}$. Let $\subst{X}{i}{r}$ denote the vector obtained by replacing the $i$-th entry of $X$ with the bit $r$. 
    
    Consider the  pair of random variables $(X,\cC(X))$. Because $\cC$ is $(\eps,\delta)$-differentially private, this is distributed similarly to the pair $(\subst{X}{i}{R},\cC(X))$, where $R$ is a uniformly random bit independent of the other values. Specifically, for any event $E \subseteq \bit^N \times \bit^N$,  $$\Pr[(\subst{X}{i}{R},\cC(X)) \in E] \leq e^{\eps}\Pr[(X,\cC(X)) \in E] + \delta. $$
    Applying this inequality to the event $E_i = \{(x,y): x_i \neq y_i \}$ shows that
    $$   \tfrac 1 2 = \Pr[\cC(X)_i\neq R]\leq e^{\eps} \Pr[\cC(X)_i\neq X_i] + \delta \, \quad \text{and thus} \quad 
    \Pr[\cC(X)_i\neq X_i] \geq e^{-\eps}(\tfrac 1  2 -\delta)\, .$$
    
    The  Hamming distance $\|\cC(X)-X\|_1$ is the  sum of the indicator random variables for the events $\cC(X)_i\neq X_i$. By linearity of expectation, the expected Hamming distance is at least $e^{-\eps}\paren{\frac 1 2 -\delta}N$.
\end{proof}

Finally, we use~\cref{lem:answering-queries,lemma:failure-to-catch,lem:reconstruct,lem:anindya-exp} 
to compete the proof of the main theorem.
\begin{proof}[Proof of Theorem~\ref{thm:noninteractiv-lb}]
{We set $\gamma=\frac 1 9.$}
Assume towards a contradiction that for some $\eps$ and $\delta$ as in the statement of the theorem, there exists an $(\eps,\delta)$-LEDP algorithm $\alg$ that for every ${3n}$-node graph approximates the number of triangles in the  graph up to additive error {$\alpha=\frac{\sqrt{\gamma} n^2}{20}$} with probability   at least $1-\frac{\gamma^2}{{9\cdot 128}}=1-\frac 1 {3^6\cdot 2^{7}}$. 
Then by \cref{lem:answering-queries}, there exists a $(2\eps,2\delta)$-DP algorithm $\cB$  that, for every secret dataset $X$ and every set of $k$ outer-product queries, %
 answers inaccurately (i.e., with additive error  {more than $\frac{\sqrt \gamma n}{4}$}) on at most {$\frac{\gamma^2 k}{64}$} of the $k$ queries with probability at least $\frac 5 6$.
By \cref{lemma:failure-to-catch}, the probability that a set of $k=\frac{128n^2}{\gamma^2}$ random outer-product queries chosen by the attacker $\cC$ 
does not catch all bad datasets
is at most $\frac 1 6$.
By a union bound, with probability at least $\frac 2 3$, the attacker $\cC$ satisfies the premise of \cref{lem:reconstruct} and the set of chosen queries catches all bad datasets. Hence, with probability at least $\frac 2 3$, the attacker $\cC$ outputs a dataset $Y^*$ which coincides with $X$ on at least $(1-\gamma) n^2$ entries. The expected Hamming distance $\E[\|X-Y^*\|_1]$ is therefore at most $\frac 2 3 \cdot \gamma n^2 + \frac 1 3 n^2 = \frac{1+2\gamma}{3} n^2$. When $\gamma =\frac 1 9$, the expected distance is less than $0.41 n^2$.

Recall that the attacker $\cC$ runs $(2\eps, 2\delta)$-DP algorithm $\cB$ on a secret dataset $X$ and then post processes the output of $\cB$. Thus, $\cC$
is $(2\eps, 2\delta)$-DP,
and we can apply \cref{lem:anindya-exp} to conclude that the expected Hamming distance $\E\|\cC(X)-X\|_1$ is at most $e^{-2\eps}(\frac 1 2 -2\delta)n^2$. Since, by assumption, $\eps\leq 1/20$ and $\delta \leq 1/100$, we have $\E\|\cC(X)-X\|_1 \geq 0.43 n^2$. This contradicts the upper bound of $0.41 n^2$ above. 
\end{proof}

\ificalp

\else

\section{Noninteractive Triangle Counting}\label{sec:one-round-upper}

In this section, we prove~\cref{thm:noninteractiv-ub-intro}.
In~\cref{alg:one-round}, we state a simple algorithm 
based on Randomized Response that obtains
an unbiased estimate of the number of triangles in the graph. 
We use ${{[n]}\choose \ell}$ to represent the set of unordered $\ell$-tuples of vertices.
Let the (fixed) variable $\indicator_{\{i, j, k\}}$ be 1 if $\{i, j, k\}$ is a triangle 
in $G$ and 0 otherwise.
Similarly, let 
$\indicator_{\{i, j\}} = 1$ if $\{i, j\} \in E$ and $\indicator_{\{i, j\}} = 0$
otherwise. Then $\indicator_{\{i, j, k\}}=\indicator_{\{i, j\}}\cdot\indicator_{\{i, k\}}\cdot\indicator_{\{j, k\}}$

The main idea in our postprocessing procedure is to rescale the noisy edges returned by Randomized Response, so that the rescaled random variable $Y_{\{i, j\}}$ for each edge ${\{i, j\}}$
has expectation exactly $\indicator_{\{i, j\}}.$

\begin{algorithm}[hbtp]
    \caption{Triangle Counting via Randomized Response}%
    \label{alg:one-round}
\textbf{Input:} Graph $G = ([n], E)$
represented by an $n\times n$ adjacency matrix $\Adj$ with entries $a_{ij}$, privacy parameter $\eps > 0$.\\
\textbf{Output:} Approximate number of triangles in $G$.
\begin{algorithmic}[1]
    \For{$i = 1$ to $n$}
        \State {\bf Release} 
        $(X_{\{i, i+1\},}\dots, X_{\{i, n\}})
        \gets\rr_{\eps}([a_{i(i+1)}, \dots, a_{in}])$. \label{or-step:e-out}
    \EndFor
      \State  For all $\{i,j\} \in {[n]\choose 2}$, set
    $Y_{\{i, j\}} \gets \frac{X_{\{i, j\}} \cdot (e^{\eps} + 1) - 1}{e^{\eps} - 1}$.
    \State For all $\{i, j, k\} \in {[n] \choose 3}$, set  $Z_{\{i, j, k\}} \leftarrow Y_{\{i, j\}} \cdot Y_{\{j, k\}} \cdot Y_{\{i, k\}}$.
    \State {\bf Return} $\displaystyle\That \leftarrow \sum_{\{i,j,k\}\in {[n] \choose 3}} Z_{\{i, j, k\}}$.\label{line:return-T-RR}
\end{algorithmic}
\end{algorithm}

\cref{alg:one-round} is $\eps$-\ledp by~\cref{thm:rr-ldp}.
Next, we show that the estimate returned by our algorithm is unbiased. 

\begin{lem}\label{lem:one-round-unbiased}
    \cref{alg:one-round} returns an unbiased estimate of the number of 
    triangles in the input graph.
\end{lem}

\begin{proof}
Let all random variables be defined as in~\cref{alg:one-round}. We calculate the expectation for all of them. The random variable $X_{\{i, j\}}$ is the indicator for the presence of a noisy edge $\{i,j\}$ in the graph returned by Randomized Response.
If edge $\{i, j\}\in E$, then 
    $\expect\left[X_{\{i, j\}}\right] = \frac{e^{\eps}}{e^{\eps} + 1}$; otherwise, 
    $\expect\left[X_{\{i, j\}}\right] = \frac{1}{e^{\eps} + 1}$. 
    Now, we calculate the expectation of the rescaled random variables $Y_{\{i, j\}} = \frac{X_{\{i, j\}} \cdot 
    (e^{\eps} + 1) - 1}{e^{\eps} - 1}$. If $\{i, j\} \in E$, then the expectation is
    $\expect\left[Y_{\{i, j\}}\right] = \frac{\expect[X_{\{i, j\}}] \cdot (e^{\eps} + 1) - 1}{e^{\eps} - 1} = \frac{\left(1 - \frac{1}{e^{\eps} + 1}\right) \cdot (e^{\eps} + 1) - 1}{e^{\eps} - 1} = 1$.
    If $\{i, j\} \not\in E$, then the expectation is 
    $\expect\left[Y_{\{i, j\}}\right] = \frac{\frac{1}{e^{\eps} + 1} \cdot 
    (e^{\eps} + 1) - 1}{e^{\eps} - 1} = 0$.
    Thus,  the expected value of $Y_{\{i, j\}}$ is $\indicator_{\{i, j\}},$ as desired.

   Now, consider random variables $Z_{\{i, j, k\}}$. By definition of $Z_{\{i, j, k\}}$ and mutual independence of random variables $Y_{\{i, j\}}$, 
   $$\E[Z_{\{i, j, k\}}] = \E[Y_{\{i, j\}}] \cdot \E[Y_{\{j, k\}}] \cdot \E[Y_{\{i, k\}}]
   =\indicator_{\{i, j\}}\cdot\indicator_{\{i, k\}}\cdot\indicator_{\{j, k\}}
   =\indicator_{\{i, j, k\}}.$$ 
    Finally, we compute the expectation of our triangle estimator $\That.$ Let $T$ be the number of triangles in $G$. Observe that $T=\sum_{\substack{\{i,j,k\} \in {[n] \choose 3}}} \indicator_{\{i, j, k\}}.$
    By definition of $\That$ and the linearity of expectation,
   
    \begin{align}
        \expect\left[\That\right] &= \expect\Big[\sum_{\substack{\{i,j,k\} \in {[n] \choose 3}}} Z_{\{i, j, k\}}\Big] = \sum_{\substack{\{i,j,k\} \in {[n] \choose 3}}} \expect\Big[Z_{\{i, j, k\}}\Big]
 =\sum_{\substack{\{i,j,k\} \in {[n] \choose 3}}} \indicator_{\{i, j, k\}}
 =T,
 \label{eq:expected-triangles}
    \end{align}
completing the proof that $\That$ is unbiased.
\end{proof}

To complete the analysis of \cref{alg:one-round}, we compute the variance of the estimate it returns. %

\begin{lem}\label{lem:one-round-variance}
Given a privacy parameter $\eps\in(0,1]$ and an $n$-node input graph $G$ with $C_4$ cycles of length 4,
\cref{alg:one-round} returns an approximate triangle count $\That$ such that 
$\var{\That} = \Theta\left(\frac{\cfours}{\eps^2} + \frac{n^3}{\eps^6}\right)$.
\end{lem}

\begin{proof}
We compute the variance of all random variables defined in \cref{alg:one-round}.
First, $\mathrm{Var}[X_{\{i, j\}}] = \frac{e^{\eps}}{(e^{\eps} + 1)^2},$ since $X_{\{i, j\}}$ is a Bernoulli variable.
Then, $\mathrm{Var}[Y_{\{i, j\}}] = \frac{(e^{\eps} + 1)^2}{(e^{\eps} - 
1)^2} \cdot \mathrm{Var}[X_{\{i, j\}}] = \frac{e^{\eps}}{(e^{\eps} - 1)^2}$.
To compute
$\mathrm{Var}[Z_{\{i, j, k\}}]$, we compute $\expect[Y_{\{i, j\}}^2] = 
\mathrm{Var}[Y_{\{i, j\}}] + \expect[Y_{\{i, j\}}]^2 = \frac{e^{\eps}}{(e^{\eps} - 1)^2} + \indicator_{\{i, j\}}
= \Theta\left(\frac{1}{\eps^2}\right)$.
By definition of variance and since $Z^2_{\{i, j, k\}}$ is a product of independent random variables $Y^2_{\{i, j\}},  Y^2_{\{j, k\}}$, and $Y^2_{\{i, k\}}$, we get
\begin{align}
\mathrm{Var}[Z_{\{i, j, k\}}] &= 
\expect[Z_{\{i,j,k\}}^2]-\expect[Z_{\{i,j,k\}}]^2=
\expect[Y_{\{i, j\}}^2] \cdot 
\expect[Y_{\{j, k\}}^2] \cdot \expect[Y_{\{i, k\}}^2]-\indicator_{\{i,j,k\}}^2 = \Theta\left(\frac{1}{\eps^6}\right).
\end{align}

To compute $\mathrm{Var}[\That]$, we consider the following change of variables:
$U_{\{i, j, k\}} = Z_{\{i, j, k\}} - \indicator_{\{i, j, k\}}$. 
Note that $\expect[U_{\{i,j,k\}}]=0$, and that 
since $\indicator_{\{i, j, k\}}$ is fixed, $\var{U_{\{i, j, k\}}} = \var{Z_{\{i, j, k\}}}$. 
Similarly, $\var{\That - T} = \var{\That},$ since $T$ is fixed. Thus, 
\begin{align}\label{eq:variance-of-T-hat}
\var{\That} 
&= \var{\That - T}
= \var{\sum_{\{i,j,k\} \in {[n] \choose 3}} \left(Z_{\{i, j, k\}} - \indicator_{\{i, j, k\}}\right)}
= \var{\sum_{\{i,j,k\} \in {[n] \choose 3}} U_{\{i, j, k\}}}.
\end{align}
To analyze $\var{\sum_{\{i,j,k\} \in {[n] \choose 3}} U_{\{i, j, k\}}}$, observe that the covariance of two variables $U_{\{i,j,k\}}$ and $U_{\{r,q,s\}}$ is zero if the triples $\{i,j,k\}$ and $\{r,q,s\}$ intersect in at most  one vertex (since the corresponding triangles share no edges).
Let $P_2$ be the set of pairs
of variables $U_{\{i, j, k\}}$ whose index triples  share exactly two nodes. 
Then by \eqref{eq:variance-of-T-hat},
\begin{align}
\var{\That} &= \sum_{\substack{\{i,j,k\} \in {[n] \choose 3}}} \var{U_{\{i, j, k\}}} + \sum_{(U_{\{i, j, k\}}, U_{\{j, k, l\}}) \in P_2} \expect\left[U_{\{i, j, k\}} \cdot U_{\{j, k, l\}}\right]\nonumber\\
&= \sum_{\substack{\{i,j,k\} \in {[n] \choose 3}}} \Theta\left(\frac{1}{\eps^6}\right) + \sum_{(U_{\{i, j, k\}}, U_{\{j, k, l\}}) \in P_2} \expect[(Z_{\{i, j, k\}} -\indicator_{\{i, j, k\}})(Z_{\{j, k, l\}} - \indicator_{\{j, k, l\}})] \nonumber\\
&= \Theta\left(\frac{n^3}{\eps^6}\right)
+ \sum_{(U_{\{i, j, k\}}, U_{\{j, k, l\}}) \in P_2} \left(\expect[Y_{\{i, j\}} \cdot Y_{\{j, k\}}^2 \cdot Y_{\{i, k\}} \cdot Y_{\{l,j\}} \cdot Y_{\{l, k\}}] - \indicator_{\{i, j, k\}}\indicator_{\{j, k, l\}}\right) \nonumber\\
&= \Theta\left(\frac{n^3}{\eps^6}\right)
+ \sum_{(U_{\{i, j, k\}}, U_{\{j, k, l\}}) \in P_2} \left(\indicator_{\{i, j\}} \cdot \expect[Y_{\{j, k\}}^2] \cdot \indicator_{\{i, k\}} \cdot \indicator_{\{l,j\}} \cdot \indicator_{\{l, k\}} - \indicator_{\{i, j, k\}}\indicator_{\{j, k, l\}}\right) \nonumber\\
&= \Theta\left(\frac{n^3}{\eps^6} + \frac{\cfours}{\eps^2}\right) \label{eq:triangles-last}.
\end{align}
The $\cfours$ term of~\eqref{eq:triangles-last} comes from the fact that $\expect[Y_{\{j,k\}}^2] = \Theta\left(\frac{1}{\eps^2}\right)$, and that 
$\indicator_{\{i, j\}} \cdot \indicator_{\{i, k\}} \cdot \indicator_{\{l,j\}} \cdot \indicator_{\{l, k\}} = 1$
  iff $(i, j, k, l)$ is a cycle of length four. 
\Cref{eq:triangles-last} gives our final bound.
\end{proof} 

\begin{proof}[Proof of~\cref{thm:noninteractiv-ub-intro}]
The main statement in \cref{thm:noninteractiv-ub-intro} follows from \cref{lem:one-round-unbiased,lem:one-round-variance}; the statement about
the additive error with constant
probability is obtained from the variance bound.
\end{proof}
\fi

\section{The Interactive  Lower Bound}\label{sec:interactive-lb}

In this section, we present an $\Omega\Big(\frac{n^{3/2}}{\eps}\Big)$ lower bound on the additive error of every $\eps$-\ledp algorithm for  
estimating the number of triangles in a graph, 
stated formally in \cref{thm:interactive-lb-intro}.
We reduce from the problem of computing the summation in the LDP model.

\begin{defnt}[Summation function]\label{def:sum}
Let $SUM_{n}$ be the following function. For all $x_1,\dots,x_n\in\{0,1\},$
$
SUM_{n}\left(x_1, \dots, x_n\right) = \sum_{i=1}^n x_i.
$
\end{defnt}

This problem was shown to have an additive 
error lower bound of $\Omega(\sqrt{n}/\eps)$ \cite[Theorem 5.3 of arxiv v2]{JosephMNR19}. We  substitute $\alpha = \alpha_0/n$ and $\beta = \eps \alpha_0/n$ to obtain the following lemma.

\begin{lem}[\cite{%
chan2012optimal,beimel2008distributed,JosephMNR19}]\label{lem:sum_lb}
There exists a constant $c>0$ such that for every $\eps  \in (0,1)$,  $n \in \naturals, \alpha_0\in(0,n] $ and  $\delta \in \left[0, \frac{1}{10^5} \cdot \frac{\eps^3 \alpha_0^2}{n^3 \ln(n^2/\eps \alpha_0)}\right]$, if $\cB$ is an $(\eps,\delta)$-LDP algorithm where each party $i$ receives input $x_i\in \bit{}$ and
$\cB$  estimates $SUM_{n}$
up to additive error $\alpha_0$
with probability at least $2/3$,  then %
$\alpha_0 \geq c\cdot\sqrt{n}/\eps$.
\end{lem}

\begin{proof}[Proof of \cref{thm:interactive-lb-intro}]
We reduce from $SUM_n$ in the local model {and then invoke the lower bound for $SUM_n$ from \cref{lem:sum_lb}. Fix the settings of parameters $\eps>0, n\in\naturals,\delta\in(0,1)$ and consider any $(\eps,\delta)$-LEDP algorithm $\alg$ that estimates the number of triangles in $3n$-node graphs and achieves additive error $\alpha$ with probability at least $\frac 23$. We use $\alg$ as a black box to obtain an $(\eps,\delta)$-LDP algorithm $\cB$ for $SUM_n$ that achieves additive error $\alpha_0=\frac \alpha n$ with probability at least $\frac 23$.

Now we describe algorithm $\cB$. } Given an instance of $SUM_n$, where each local party holds one bit $X_i$ of the vector $(X_1,\dots,X_n),$ the parties implicitly create the following graph $G$. The vertex set consists of two sets of nodes, $V_1$ and $V_2$, where $V_1$ has size $n$ and $V_2$ has size $2n$. The nodes in $V_1$ do not have any secret information and can be simulated by any local party. The nodes in $V_2$ are $[2n]$,  and each party $i\in [n]$ is responsible for simulating nodes $2i-1$ and $2i$ in $V_2$. 
To create the edges of $G,$ we add edges of the complete bipartite graph between $V_1$ and $V_2$. In addition, each pair of nodes $(2i-1,2i)$ in $V_2$ has an edge between them if and only
if $x_i=1$. See~\Cref{fig:interactive-lb} for an illustration.
The parties simulate algorithm $\alg$ to obtain an estimate $\hat T$ for the number of triangles in $G$. The output of algorithm $\cB$ is then set to $\hat S=\frac{\hat T} n$.

Next, we analyze the accuracy and privacy of the reduction algorithm $\cB$.
Let $S=x_1+\ldots +x_n.$ Observe that any triangle in $G$ must have two vertices in $V_2$ and an edge between a pair of matched nodes. Any such edge contributes exactly $n$ triangles. So, the total number of triangles in $G$ is $T=Sn.$
Thus, an estimate $\hat T$ for the number of triangles in $G$ with additive error $\alpha$ results in an estimate $\hat S=\frac{\hat T} n$ with additive error $\alpha_0=\frac \alpha n$ for the original instance of $SUM_n$.
Moreover, since algorithm $\alg$ is $(\eps,\delta)$-LEDP,  algorithm $\cB$ is $(\eps,\delta)$-LDP with respect to the secret dataset $X$.

By the lower bound on the accuracy of local algorithms for $SUM_n$ stated in \cref{lem:sum_lb}, 
there exists a constant $c>0$ such that for every $\eps  \in (0,1)$,  $n \in \naturals, \alpha_0\in(0,n] $ and  $\delta \in \left[0, \frac{1}{10^5} \cdot \frac{\eps^3 \alpha_0^2}{n^3 \ln(n^2/\eps \alpha_0)}\right]$, we have $\alpha_0 \geq c\cdot\sqrt{n}/\eps$. 
Substituting, $\alpha_0=\frac \alpha n$, we get that the range of $\delta$ is $\left[0,\frac{1}{10^5}\cdot \frac{\eps^2\cdot \alpha^2}{n^5 \ln(n^3/\eps\alpha))}\right]$.
That is, there exists a constant $c>0$ such that for every $\eps  \in (0,1)$,  $n \in \naturals, \alpha_0\in(0,n^2] $ and  $\delta\in\left[0,\frac{1}{10^5}\cdot \frac{\eps^2\cdot \alpha^2}{n^5 \ln(n^3/\eps\alpha))}\right]$,
every $(\eps,\delta)$-LEPD algorithm that estimates the number of triangles in $3n$-node graphs with additive error $\alpha$ with probability at least $\frac 2 3$ must satisfy $\alpha\geq c\cdot \frac{n^{3/2}}\eps.$
Since \cref{thm:interactive-lb-intro} is stated for $n$-node graphs (as opposed to $3n$-node graphs), we can substitute $n/3$ instead of $n$ in the upper bound on the applicable values of $\delta$, increasing the constant in the upper bound to $\frac{243}{10^5}>\frac 1 {500}$.
 
\end{proof}

\begin{figure}
    \begin{center}
        \begin{tikzpicture}
	    \begin{scope}[xscale=0.8, yscale=0.6]
    \interLB
    \end{scope}
    \end{tikzpicture}
    \end{center}
    \caption{An instance of the interactive $\Omega(n^{3/2})$ lower bound consists of a complete bipartite graph with parts $V_1, V_2$ of sizes $n$ and $2n$, respectively;
  in addition, there is an edge between each pair $\{2i-1, 2i\}$ iff the secret input bit $X_i=1$.
    }
    
    \label{fig:interactive-lb}
\end{figure}

\paragraph{Acknowledgments.} T.E.\ was supported by the NSF TRIPODS program, award DMS-2022448, and Boston University.  This work was partially done while T.E.\ was affiliated with  Boston University and MIT. S.R.\ was supported in part by the NSF award DMS-2022446.
A.S.\ was supported in part by NSF awards CCF-1763786 and CNS-2120667 as well as Faculty Awards from Google and Apple.
We thank Iden Kalemaj and Satchit Sivakumar for helpful comments on the initial version of our results. 

\bibliographystyle{alpha}
\bibliography{ref}

\appendix
\iftpdp
\section{Additional Preliminaries}\label{sec:append-privacy-tools}
\else
\section{Additional Tools}\label{sec:append-privacy-tools}
\fi
\subsection{Privacy Tools}\label{sec:privacy-tools}
Differential privacy is preserved under composition and postprocessing. 

\begin{theorem}[Composition Theorem, adapted from~\cite{DMNS06,DL09,DRV10}]\label{thm:composition}
Let $\eps>0$ and $\delta\in[0,1)$.
    If algorithm $R$ runs  $(\eps,\delta)$-local randomizers $R_1, \dots, R_k$, then $R$ is a $(k\eps,k \delta)$-local randomizer.
\end{theorem}

The next theorem states that the result of postprocessing  the output
of DP algorithm is DP. 
\begin{theorem}[Postprocessing~\cite{DMNS06,BS16}]\label{thm:post-processing}
Let $\alg$ be an $(\eps,\delta)$-DP algorithm and $\cP$ be an arbitrary (randomized)
mapping from $\range(\alg)$ to an arbitrary
set. The algorithm that runs $\cP$ on the output of $\alg$ is $(\eps,\delta)$-DP.
\end{theorem}

For completeness, we state the definition of global sensitivity used in our discussion of the central model.
\begin{definition}[Global Sensitivity~\cite{DMNS06}]\label{def:global-sensitivity}
Let $d\in\mathbb{N}$ and $f: \domain \rightarrow \reals^d$ be a function on some domain~$\domain$. The 
    \defn{$\ell_1$-{sensitivity}} of 
    $f$ is $\df = \max_{x, x'} \norm{f(x) 
    - f(x')}_1$, where the maximum is over all neighboring $x, x' \in \domain$, for 
    some appropriate definition of 
    neighboring.
\end{definition}

\subsection{Concentration Bounds}\label{sec:concentration-bounds}
\begin{theorem}[Chernoff Bound~\cite{chernoff1952measure}]\label{prel:chernoff}
Let $X_1, \ldots, X_n$ be independent random variables in 
$\{0,1\}.$ Let $X=\sum_{i=1}^n X_i$ and $\mu=\E[X]$. Then for any $\delta>0$, 
\[
\Pr\left[X<(1-\delta)\mu\right]\leq \exp\left(-\frac{\delta^2\cdot \mu}{2} \right).
\]

\end{theorem}

\begin{theorem}[Paley–Zygmund Inequality \cite{paley1932note}]\label{prel:paley-zygmund}
If $Z\geq 0$ is a random variable with finite variance, then, for all $0\leq \theta\leq 1,$
\[
\Pr[Z>\theta \cdot \expect[Z]]\geq (1-\theta^2)\frac{\expect[Z]^2}{\expect[Z^2]}.
\]
\end{theorem}

\end{document}